\documentclass[10pt, doublecolumn]{IEEEtran}
\usepackage{epsfig,latexsym}
\usepackage{float}
\usepackage{indentfirst}
\usepackage{amsmath}
\usepackage{amssymb}
\usepackage{times}
\usepackage{subfigure}
\usepackage{psfrag}
\usepackage{hyperref}
\usepackage{cite}
\usepackage{lastpage}
\usepackage{fancyhdr}
\usepackage{color}
 \usepackage{amsthm}
\usepackage{bigints}
\newtheorem{Lemma}{Lemma}
\newtheorem{lemma}[Lemma]{$\mathbf{Lemma}$}

\begin{document}
\title{ {\huge  A Simple Design of  IRS-NOMA Transmission      }}

\author{ Zhiguo Ding     and H. Vincent Poor, \thanks{ 
  
\vspace{-1em}
    Z. Ding and H. V. Poor are  with the Department of
Electrical Engineering, Princeton University, Princeton, NJ 08544,
USA. Z. Ding
 is also  with the School of
Electrical and Electronic Engineering, the University of Manchester, Manchester, UK (email: \href{mailto:zhiguo.ding@manchester.ac.uk}{zhiguo.ding@manchester.ac.uk}, \href{mailto:poor@princeton.edu}{poor@princeton.edu}).

  }\vspace{-2em}}
 \maketitle

\begin{abstract} 
  This letter proposes a simple design of  intelligent reflecting surface (IRS) assisted  non-orthogonal multiple access (NOMA)  transmission, which can ensure that more users are served on each orthogonal   spatial direction than spatial division multiple access (SDMA). In particular, by employing IRS, the directions of   users'   channel vectors can be effectively aligned, which facilitates the implementation of NOMA. Both analytical and simulation results are provided to demonstrate the performance of the proposed  IRS-NOMA scheme and also study the impact of hardware impairments on  IRS-NOMA. 
\end{abstract}
  
\vspace{-1.6em}
\section{Introduction}
Non-orthogonal multiple access (NOMA) has been recognized as a promising multiple access candidate   for future mobile networks \cite{nomama}. The key idea of NOMA is to serve multiple users on each orthogonal bandwidth resource block.   In   scenarios with multiple-antenna nodes, orthogonal spatial directions can be viewed as a type of resource blocks. Conventional orthogonal multiple access (OMA), such as spatial division multiple access (SDMA), is to serve a single user on each spatial direction, whereas the use of NOMA is to ensure that multiple users are served simultaneously  on each spatial direction and hence improves spectral efficiency. However,  it is important to point out that the use of NOMA is not always preferable    \cite{7555306}. For example,  if users' channel vectors are orthogonal to each other,   SDMA is more preferable than NOMA, whereas the situation, in which the directions of the users' channel vectors are the same,  is the ideal case for the implementation of NOMA.  

Therefore, an important question to broaden  the applications of NOMA is whether the directions of users' channel vectors can be manipulated, i.e., aligning one user's channel with the others'. This is difficult in  conventional wireless systems, since  the   users' channels  are fixed and determined by propagation environments.  Motivated by this difficulty, this letter is to propose a new type of NOMA transmission by employing   the intelligent reflecting surface (IRS) which can be viewed as  a low-cost antenna array consisting of a large number of
reconfigurable reflecting  elements \cite{irs1, irs2}. By applying IRS, the direction of a user's channel vector can be effectively tuned,  which facilitates the implementation of NOMA. In particular, the spectral efficiency and connectivity can be  improved by IRS-NOMA since a single spatial direction can be used to serve multiple users, even if their original channels are not aligned.  Due to the low-cost feature of IRS, e.g., finite-resolution  phase shifters,   a user's channel vector cannot be accurately aligned to a target direction. The impact of this hardware impairment on  IRS-NOMA is investigated and the performance of the developed practical IRS-NOMA transmission scheme is characterized in this letter.

\begin{figure}[t]\vspace{-0.5em}
 \centering
         \includegraphics[width=.24\textwidth]{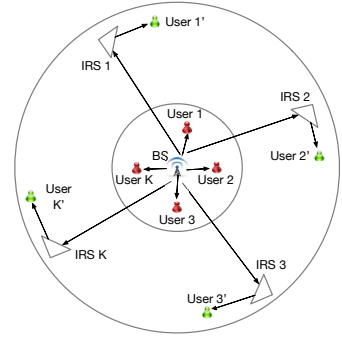} \vspace{-0.5em}
\caption{ A system diagram for IRS-NOMA with $2K$  users and $K$ IRS's.    }\label{fig1x} \vspace{-2em}
\end{figure}

\vspace{-0.5em}
\section{System Model}
Consider a multi-user downlink   scenario as shown in Fig.~\ref{fig1x}. There are two types of users, namely near users and cell-edge users, where it is assumed that there is no direct link between the base station and the cell-edge users.   SDMA is used as a benchmark for IRS-NOMA, where $K$ near users  are served by the base station which is equipped with $M$ ($M\geq K$) antennas and performs zero forcing beamforming.     To better illustrate the benefit of   IRS-NOMA, we further assume that the $K$ near users are scheduled because their channel vectors  are orthogonal to each other, which means that the beamforming vectors, denoted by $\mathbf{w}_k$, $1\leq k \leq K$, are orthonormal vectors.

The main benefit for employing  IRS-NOMA     is to ensure that more users are connected than conventional SDMA. For illustration purposes, we assume that  on each beam, $\mathbf{w}_k$,  one additional user, denoted by user  $k'$, is served with the help of an IRS which is equipped with $N$ antennas, as shown in Fig.~\ref{fig1x}. In addition, we assume that  only user $k'$ can hear   IRS $k$ since the IRS's are deployed close to the cell-edge  users\footnote{These assumptions facilitate the illustration for  the benefit of IRS-NOMA. How to extend the obtained analytical results to more general scenarios is an important topic for future research but beyond the scope of this letter. }

The base station broadcasts  $\sum^{K}_{k=1} \mathbf{w}_k(\alpha_1s_k+\alpha_2s_{k'})$, where $s_{k'}$ denotes the signal to be sent to user $k'$, $s_k$ is the signal to be sent to user $k$, $\alpha_i$ denotes the power allocation coefficient, and $\alpha_1^2+\alpha_2^2=1$.  Because it is assumed that there is no direct link between the IRS's and the near users, the performance analysis for user $k$ is exactly the same as those in conventional NOMA systems, and hence in this letter we   focus on the performance at user $k'$ only.

Therefore, the   signal received by user $k'$ is given by
\begin{align}
y_{k'} = \mathbf{h}^H_{k'} \boldsymbol \Theta_k \mathbf{G}_k \sum^{K}_{k=1}\mathbf{w}_k (\alpha_1s_k+\alpha_2s_{k'})+w_{k'},
\end{align}
where $\mathbf{G}_k$ denotes the $N\times M$ complex Gaussian channel matrix from the base station to the IRS associated with user $k'$,   $\mathbf{h}_{k'}$ denotes the  complex Gaussian    channel vector from   the  IRS to user $k'$, and $w_{k'}$ denotes the noise.   $\boldsymbol \Theta_k $ is a diagonal matrix, and    each of its main diagonal elements is denoted by $\beta_{k,i}e^{-j\theta_{k,i}}$, where $\theta_{k,i}$ denotes the reflection phase shift and $\beta_{k,i}$ denotes  the amplitude reflection coefficient  \cite{irs1, irs2}.

As in conventional NOMA, it is assumed that $\alpha_1\leq \alpha_2$, and hence the signal-interference-plus-noise (SINR) for user $k'$ to decode its message  is given by 
\begin{align}\nonumber
\text{SINR}_{k'} =& \frac{ |\mathbf{h}^H_{k'} \boldsymbol \Theta_k \mathbf{G}_k \mathbf{w}_k|^2 \alpha_1^2}{ |\mathbf{h}^H_{k'} \boldsymbol \Theta_k \mathbf{G}_k \mathbf{w}_k|^2 \alpha_2^2+ \sum^{K}_{i=1, i\neq k}|\mathbf{h}^H_{k'} \boldsymbol \Theta_k \mathbf{G}_k\mathbf{w}_i|^2  +\frac{1}{\rho}}\\ \label{SINR}
=&
 \frac{ | \boldsymbol \theta_k^H \mathbf{D}_{k'} \mathbf{h}_k|^2 \alpha_1^2}{| \boldsymbol \theta^H _k\mathbf{D}_{k'} \mathbf{h}_k|^2  \alpha_2^2+ \sum^{K}_{i=1,i\neq k}|\boldsymbol \theta_k^H \mathbf{D}_{k'} \mathbf{h}_i|^2  +\frac{1}{\rho}},\hspace{-1em}
\end{align}
where $\rho$ denotes the transmit signal-to-noise ratio (SNR), $\mathbf{h}_i=\mathbf{G}_k\mathbf{w}_i$, $ \boldsymbol \theta_k$ is an $N\times 1$ vector containing the elements on the main diagonal of $ \boldsymbol \Theta_k^H$, and $\mathbf{D}_{k'}$ is a diagonal matrix with its diagonal elements obtained from  $\mathbf{h}_{k'}^H$.

\section{Designs of IRS-NOMA }
\subsection{   IRS-NOMA with Ideal Beamforming}
As shown in \eqref{SINR}, the performance of IRS-NOMA is determined by the choice of $\boldsymbol \Theta_k$. Ideal designs of IRS-NOMA typically require beamforming with infinite resolution, i.e., the hardware circuit can support arbitrary choices for the phase shift and amplitude coefficient, $\theta_{k,i}$ and $\beta_{k,i}$. Take the zero-forcing design as an example. 
In order to suppress inter-pair interference, the use of zero-forcing beamforming implies that $\boldsymbol \theta_k $ should satisfy the following constraints:
\begin{align}\label{zf}
\boldsymbol \theta_k^H \mathbf{D}_{k'} \mathbf{h}_i =0,
\end{align}
for $i\neq k$.   Denote $\mathbf{V}_k$ by an $N\times (N-K+1)$ matrix  collecting the basis vectors of the null space of $\begin{bmatrix}\mathbf{D}_{k'} \mathbf{h}_1& \cdots& \mathbf{D}_{k'} \mathbf{h}_{i-1} &\mathbf{D}_{k'} \mathbf{h}_{i+1}&\cdots&\mathbf{D}_{k'} \mathbf{h}_K\end{bmatrix}$. Therefore,    $\boldsymbol \theta_k $ can be obtained as $\mathbf{V}_k\mathbf{x}$, and the optimal $\mathbf{x}$ can be  obtained as follows:
\begin{align} 
\rm{max.} &\qquad  \frac{ | \mathbf{x}^H \mathbf{V}_k^H\mathbf{D}_{k'} \mathbf{h}_k|^2 \alpha_1^2}{|  \mathbf{x}^H \mathbf{V}_k^H \mathbf{D}_{k'} \mathbf{h}_k|^2  \alpha_2^2   +\frac{1}{\rho}}\\
\rm{s.t.} & \qquad|\mathbf{x}|^2\leq 1.
\end{align}
By using the fact that $ \frac{\alpha_1^2 y}{\alpha_2^2y+\frac{1}{\rho}}$ is a mono-increasing function of $y$ and also applying Cauchy-Schwarz inequality,  the maximum of the SINR is     achieved by $  \boldsymbol \theta_k^*=\mathbf{V}_k\frac{\mathbf{V}_k^H \mathbf{D}_{k'} \mathbf{h}_k}{|\mathbf{V}_k^H \mathbf{D}_{k'} \mathbf{h}_k|}$. Evidently, such an ideal design requires that the number of  possible choices for the phase shift and the amplitude, $\theta_{k,i}$ and $\beta_{k,i}$, is infinite. We note that the other ideal designs, e.g., directly maximizing the SINR in \eqref{SINR}, lead to the same conclusion.

\subsection{IRS-NOMA With Finite Resolution Beamforming}
 In practice, the choices for $\theta_{k,i}$ and  $\beta_{k,i}$ cannot be arbitrary due to the hardware limitations. A straightforward   design for IRS-NOMA with finite resolution beamforming is inspired by lens antenna arrays in millimeter-wave  networks \cite{7416205}. In particular, denote an $N\times N$ discrete Fourier transform (DFT) matrix by $\mathbf{F}_n$, and the optimal $\boldsymbol \theta_k$ to maximize the SINR in \eqref{SINR} can be found   by an  exhaustive search among the columns of  $\mathbf{F}_N$, which can be realized by finite-resolution phase shiters. 
 
An alternative  low-cost implementation is to apply on-off control to IRS-NOMA, i.e., each diagonal element of $\boldsymbol \Theta$ is either $0$ (off) or $1$ (on).   Without loss of generality, assume that $N=PQ$, where $P$ and $Q$ are integers. Define $\mathbf{V}= \frac{1}{\sqrt{Q}}\mathbf{I}_P\otimes \mathbf{1}_Q$, where $\mathbf{I}_P$ is a $P\times P$ identity matrix, $\mathbf{1}_Q$ is a $Q\times 1$ all-ones vector, and $\otimes$ denotes the Kronecker product. Denote $\mathbf{v}_p$ by the $p$-th column of $\mathbf{V}$, where it is easy to show that $\mathbf{v}_p^H \mathbf{v}_l=0$ for $p\neq l$, and $\mathbf{v}_p^H  \mathbf{v}_p=1$.     $\boldsymbol \theta_k$ is selected based  on the following criterion:
\begin{align}\label{on off criterion}
\underset{ \mathbf{v}_p}{\max} ~ \frac{ | \mathbf{v}_p^H \mathbf{D}_{k'} \mathbf{h}_k|^2 \alpha_1^2}{|  \mathbf{v}_p^H \mathbf{D}_{k'} \mathbf{h}_k|^2  \alpha_2^2+ \sum^{K}_{i=1,i\neq k}| \mathbf{v}_p^H \mathbf{D}_{k'} \mathbf{h}_i|^2  +\frac{1}{\rho}}.
\end{align}

As shown in the remaining of the letter, the use of on-off control not only yields better performance than the DFT-based design, but also ensures that insightful analytical results can be developed.  
\begin{lemma}\label{lemma1}
For the single user case ($K=1$), the use of IRS-NOMA with on-off control can achieve   the following outage probability   at user $k'$:
\begin{align}\label{lemma pro1}
\mathrm{P}_{k'} = &  \frac{\xi^{\frac{N}{2}}}{(\Gamma(Q))^P} \left(
 \xi^{-\frac{Q}{2}}\Gamma(Q) -2 K_Q\left(2\xi^{\frac{1}{2}}\right)
\right)^P,
\end{align}
if $\alpha_1^2-\alpha_2^2\epsilon_{k'}>0$, otherwise $\mathrm{P}_{k'} =1$,  where $\xi=\frac{Q\epsilon_{k'}}{\rho(\alpha_1^2-\alpha_2^2\epsilon_{k'})}$, $\epsilon_{k'}=2^{R_{k'}}-1$, $R_{k'}$ denotes the target rate of user $k'$, $K_n(\cdot)$ denotes the   modified Bessel function of the second kind, and $\Gamma(\cdot)$ denotes the gamma function.
At high SNR, the outage probability can be approximated as follows:
  \begin{align}\label{lemma pro2}
\mathrm{P}_{k'}  \approx &\left\{\begin{array}{ll} 
    \xi^N\left(-\ln(\xi) \right)^N,&Q=1 \\    \frac{\xi^P}{(Q-1)^P}, &Q\geq 2\end{array}\right.,
\end{align}
for $\alpha_1^2-\alpha_2^2\epsilon_{k'}>0$. 
\end{lemma}
\begin{proof}
For the   case $K=1$ and $\boldsymbol \theta_k=\mathbf{v}_p$,   $\text{SINR}_{k'} $ in \eqref{SINR} can be simplified as follows:
\begin{align} 
\text{SINR}_{k',p} =&  
 \frac{ | \mathbf{v}_p^H \mathbf{D}_{k'} \mathbf{h}_k|^2 \alpha_1^2}{| \mathbf{v}_p^H \mathbf{D}_{k'} \mathbf{h}_k|^2  \alpha_2^2  +\frac{1}{\rho}}.
\end{align}
Because of the structure of $\mathbf{v}_p$, $\sqrt{Q} \mathbf{v}_p^H \mathbf{D}_{k'} \mathbf{h}_k$ is simply an inner product of two $Q\times 1$ complex Gaussian vectors.   By first treating  $\sqrt{Q} \mathbf{v}_p^H \mathbf{D}_{k'} \mathbf{h}_k$ as a complex Gaussian  random variable with zero mean and variance $| \mathbf{h}_k|^2$ and using the fact that $| \mathbf{h}_k|^2$ is chi-square distributed, the probability density function (pdf) of $\sqrt{Q} \mathbf{v}_p^H \mathbf{D}_{k'} \mathbf{h}_k$ can be obtained  as follows:
\begin{align}\label{pdf q}
f_{Q| \mathbf{v}_p^H \mathbf{D}_{k'} \mathbf{h}_k|^2 }(x) = \frac{2x^{\frac{Q-1}{2}}}{\Gamma(Q)}K_{Q-1}(2\sqrt{x}),
\end{align}
where the details for the derivation can be found in \cite{6725568,m343}. 

Therefore, for the case $\boldsymbol \theta_k=\mathbf{v}_p$, the outage probability can be expressed as follows:
\begin{align}
\mathrm{P}_{k',p} = \mathrm{P}\left(\log(1+\text{SINR}_{k',p} )<R_{k'}\right).
\end{align}
By applying the pdf shown in \eqref{pdf q}, the outage probability can be expressed as follows:
\begin{align}
\mathrm{P}_{k',p} = & \int^{\xi}_{0} f_{Q| \mathbf{v}_p^H \mathbf{D}_{k'} \mathbf{h}_k|^2 }(x) dx
\\\nonumber
= & \frac{2}{\Gamma(Q)}\int^{\xi}_{0} x^{\frac{Q-1}{2}}K_{Q-1}(2\sqrt{x}) dx
\\\nonumber
= & \frac{1}{\Gamma(Q)} \xi^{\frac{Q+1}{2}}\left(
 \xi^{-\frac{Q+1}{2}}\Gamma(Q) -2\xi^{-\frac{1}{2}}K_Q\left(2\xi^{\frac{1}{2}}\right)
\right),
\end{align}
where  the last step follows from Eq. (6.561.8) in \cite{GRADSHTEYN}.

For IRS-NOMA with on-off control, $\mathbf{V}= \frac{1}{\sqrt{Q}}\mathbf{I}_P\otimes \mathbf{1}_Q$, and hence one can easily verify that $ \mathbf{v}_p^H \mathbf{D}_{k'} \mathbf{h}_k$ and $ \mathbf{v}_l^H \mathbf{D}_{k'} \mathbf{h}_k$ are independent and identically distributed (i.i.d.)  for $p\neq l$. Therefore, the use of the selection criterion in \eqref{on off criterion} ensures that the outage probability at user $k'$ can be expressed as follows:
\begin{align}
\mathrm{P}_{k'} = &  \frac{\xi^{\frac{P(Q+1)}{2}}}{(\Gamma(Q))^P} \left(
 \xi^{-\frac{Q+1}{2}}\Gamma(Q) -2 \xi^{-\frac{1}{2}}K_Q\left(2\xi^{\frac{1}{2}}\right)
\right)^P.
\end{align}
With some algebraic manipulations,  \eqref{lemma pro1} in the lemma can be obtained. 

In order to find the high SNR approximation for \eqref{lemma pro1}, we first note that at high SNR, $\rho\rightarrow \infty$, which means that $\xi\rightarrow 0$. Recall that $K_n(z)$ can be approximated as follows: \cite{GRADSHTEYN}
\begin{align}\label{app1}
K_n(z) \approx \frac{1}{2} \left(\frac{(n-1)!}{\left(\frac{z}{2}\right)^n}-\frac{(n-2)!}{\left(\frac{z}{2}\right)^{n-2}}\right),
\end{align}
for $n\geq 2$ and  $z\rightarrow 0$.  Therefore, the outage probability $\mathrm{P}_{k'} $ can be approximated as follows: 
 \begin{align}\nonumber
\mathrm{P}_{k'} = &  \frac{1}{(\Gamma(Q))^P}  \left(
  \Gamma(Q) -2 \xi^{\frac{Q}{2}}K_Q\left(2\xi^{\frac{1}{2}}\right)
\right)^P
\\\nonumber \approx &
 \frac{1}{(\Gamma(Q))^P}  \left(
  \Gamma(Q) -  \xi^{\frac{Q}{2}}
   \left(\frac{(Q-1)!}{\xi^{\frac{Q}{2}}}-\frac{(Q-2)!}{\xi^{\frac{Q-2}{2}}}\right) 
\right)^P
\\\label{15} \approx &
   \frac{\xi^P}{(Q-1)^P} ,
\end{align}
for the cases with $Q\geq 2$. 

For the case with $Q=1$, unlike \eqref{app1}, a different  approximation for the Bessel function  will be used as shown in the following:
\begin{align}
K_1(z) \approx \frac{1}{2}  \frac{1}{\left(\frac{z}{2}\right)}+\left(\frac{z}{2}\right) \ln \left(\frac{z}{2}\right),
\end{align}
for $z\rightarrow 0$. Therefore, for the case with $Q=1$, the outage probability can be approximated as follows:
  \begin{align}\label{17}
\mathrm{P}_{k'} = &    \left(
  1 -2 \xi^{\frac{1}{2}}K_1\left(2\xi^{\frac{1}{2}}\right)
\right)^N
\\\nonumber \approx &
   \left(
 1 -  \xi^{\frac{1}{2}}
   \left(\frac{1}{\xi^{\frac{1}{2}}}+ \xi^{\frac{1}{2}}\ln(\xi) \right)
\right)^P
  \approx 
    \xi^N\left(-\ln(\xi) \right)^N.
\end{align}
By combining \eqref{15} with \eqref{17}, \eqref{lemma pro2} in the lemma can be obtained, and the proof for the lemma is complete. 
\end{proof}
{\it Remark 1:} The diversity gain for the case with $Q=1$ can be found as follows:
\begin{align}
-\underset{\rho\rightarrow \infty}{\lim }\frac{\log \mathrm{P}_{k'} }{\log \rho} = &\underset{\xi \rightarrow 0}{\lim }\frac{\log [    \xi^N\left(-\ln(\xi) \right)^N]   }{\log \xi}\\\nonumber= & N +N \underset{\xi \rightarrow 0}{\lim }\frac{\log [    \left(-\ln \xi \right)]   }{\log \xi}=N,
\end{align}
where the last step follows by applying L'Hospital's rule.
It is straightforward to show that the diversity gain for the case $Q\geq 2$ is $P$. Therefore,   the choice of $Q=1$ is diversity optimal to IRS-NOMA with on-off control. 

{\it Remark 2:} Lemma \ref{lemma1} is only applicable to IRS-NOMA with on-off control. The analytical results for IRS-NOMA with DFT are difficult to obtain, mainly because of the correlation between $  |\mathbf{v}_p^H \mathbf{D}_{k'} \mathbf{h}_k|^2$ and $ |\mathbf{v}_i^H \mathbf{D}_{k'} \mathbf{h}_k|^2$, for $i\neq p$. We note that simulation results indicate that this correlation is very weak, which results in an observation that the diversity order achieved by the DFT case is similar to that of the scheme with on-off control.

Note that for the multi-user case ($K\geq 2$), the outage probability achieved by IRS-NOMA with on-off control is   difficult to analyze,  due to the correlation between  $ | \mathbf{v}_p^H \mathbf{D}_{k'} \mathbf{h}_k|^2$ and $ | \mathbf{v}_p^H \mathbf{D}_{k'} \mathbf{h}_i|^2$, for $i\neq p$. Consistent to the single-user case, simulation results show that   $Q=1$ is also optimal in the high SNR regime for the multi-user case. Therefore, the choice of $Q=1$ is   focused in the following, where a closed-form expression for the outage probability can be obtained, as shown in the the following lemma. 
\begin{lemma}\label{lemma2}
For the multi-user case $K\geq 2$, the outage probability achieved by IRS-NOMA with on-off control ($Q=1$) is given by
\begin{align}  \label{19}
\mathrm{P}_{k'} = &   \left(1-     \frac{2 \sqrt{\frac{\epsilon_{k'}}{\rho \tau}} K_1\left(2\sqrt{\frac{\epsilon_{k'}}{\rho \tau}} \right)}{\left(1+\frac{\epsilon_{k'}}{\tau}\right)^{K-1}}\right)^N,
\end{align}
where $\tau = \alpha_1^2-\epsilon_{k'}\alpha_2^2$. At high SNR, the outage probability can be approximated as follows:
\begin{align}   \label{20}
\mathrm{P}_{k'} \approx &     \left(1-     \frac{1}{\left(1+\frac{\epsilon_{k'}}{\tau}\right)^{K-1}}\right)^N. 
\end{align}
\end{lemma}
\begin{proof}
 For the special  case  $Q=1$, $\mathbf{v}_p$ is an $N\times 1$ vector with all of its elements being zero except its $p$-th element being one. Therefore, $\mathbf{v}_p^H \mathbf{D}_{k'} \mathbf{h}_k$ becomes $ {d}_{k',p}  {h}_{k,p}$, where  $d_{k',p}$ is the $p$-th element on the diagonal of $\mathbf{D}_{k'}$ and $h_{k,p}$ is the $p$-th element of $\mathbf{h}_k$. Therefore, with $\boldsymbol \theta_k=\mathbf{v}_p$,  $\text{SINR}_{k'} $ can be simplified as follows:
\begin{align} \nonumber
\text{SINR}_{k',p} =&   \frac{ |d_{k',p}|^2 |h_{k,p}|^2 \alpha_1^2}{|    {d}_{k',p}|^2 | {h}_{k,p}k|^2  \alpha_2^2+|{d}_{k',p}|^2 \sum^{K}_{i=1,i\neq k}    | {h}_{i,p}|^2  +\frac{1}{\rho}}.
\end{align}
  We first note that $\mathbf{h}_k=\mathbf{G}_k\mathbf{w}_k$  is still a complex Gaussian vector, since $\mathbf{w}_k$ is normalized. We further note that $\mathbf{h}_k$ and $\mathbf{h}_i$, $k\neq i$, are independent since  $\mathbf{w}_k$  and $\mathbf{w}_i$ are assumed to be orthonormal vectors.   Therefore, $h_{k,p}$ and $h_{i,p}$ are   independent and complex Gaussian distributed.  Hence, the SINR can be further simplified as follows:
\begin{align} 
\text{SINR}_{1',p} =&   \frac{ x y \alpha_1^2}{xy \alpha_2^2+ x z  +\frac{1}{\rho}},
\end{align}
where $x= |d_{k',p}|^2 $ and $y=|h_{k,p}|^2$ are two independent and exponentially distributed random variables, and $z=\sum^{K}_{i=1,i\neq k} | {h}_{i,p}|^2$ is chi-square distributed with $2(K-1)$ degrees of freedom. 
Therefore, the outage probability can be expressed as follows:
\begin{align}\nonumber 
\mathrm{P}_{k',p} = & \int^{\infty}_{0}  \int^{\infty}_{0}\left(1 -e^{-\frac{\epsilon_{k'}xz +\frac{\epsilon_{k'}}{\rho}}{x(\alpha_1^2-\epsilon_{k'}\alpha_2^2)}}\right) \frac{e^{-x}dx}{(K-2)!}   z^{K-2}e^{-z}dz\\\nonumber
= & 1- \frac{1}{(K-2)!} 2\sqrt{\frac{\epsilon_{k'}}{\rho \tau}} K_1\left(2\sqrt{\frac{\epsilon_{k'}}{\rho \tau}} \right) \\  \nonumber&\times \int^{\infty}_{0} e^{-\frac{\epsilon_{k'}}{\tau}z} z^{K-2}e^{-z}dz,
\end{align}
where    the last step follows from Eq. (3.324.1) in \cite{GRADSHTEYN}. 

By applying Eq. (3.381.4) in \cite{GRADSHTEYN}, the outage probability when $\boldsymbol \theta_k=\mathbf{v}_p$ can be obtained as follows:
\begin{align}  
\mathrm{P}_{k',p} = & 1 -  2\sqrt{\frac{\epsilon_{k'}}{\rho \tau}} K_1\left(2\sqrt{\frac{\epsilon_{k'}}{\rho \tau}} \right)    \frac{1}{\left(1+\frac{\epsilon_{k'}}{\tau}\right)^{K-1}}. 
\end{align}
Because of  the structure of $\mathbf{V}$, the SINRs for different $\mathbf{v}_p$ are i.i.d., and therefore, the  outage probability $\mathrm{P}_{k'} $ can be obtained as shown in \eqref{19} in the lemma. 

At high SNR, by using the fact that $K_1(x)\approx \frac{1}{x}$ for $x\rightarrow 0$,  the outage probability can be approximated as follows:
\begin{align}  \nonumber
\mathrm{P}_{k'} \approx &  \left(1-     \frac{2  \sqrt{\frac{\epsilon_{k'}}{\rho \tau}}  \left(2\sqrt{\frac{\epsilon_{k'}}{\rho \tau}} \right)^{-1}}{\left(1+\frac{\epsilon_{k'}}{\tau}\right)^{K-1}}\right)^N.
\end{align}
With some algebraic manipulations, the approximation shown in \eqref{20}   can be obtained, and the lemma is proved. 
 \end{proof}
 {\it Remark 3:}  The high SNR approximation shown in \eqref{20} indicates the existence of an error floor for the outage probability, i.e., the outage probability does not go to zero by simply increasing the transmission power. However, the outage probability can be reduced by increasing $N$, as shown in \eqref{20}. 

\begin{figure}[t] \vspace{-2em}
\begin{center}\subfigure[Comparison between different IRS-NOMA ($Q=1$)]{\label{fig1a}\includegraphics[width=0.24\textwidth]{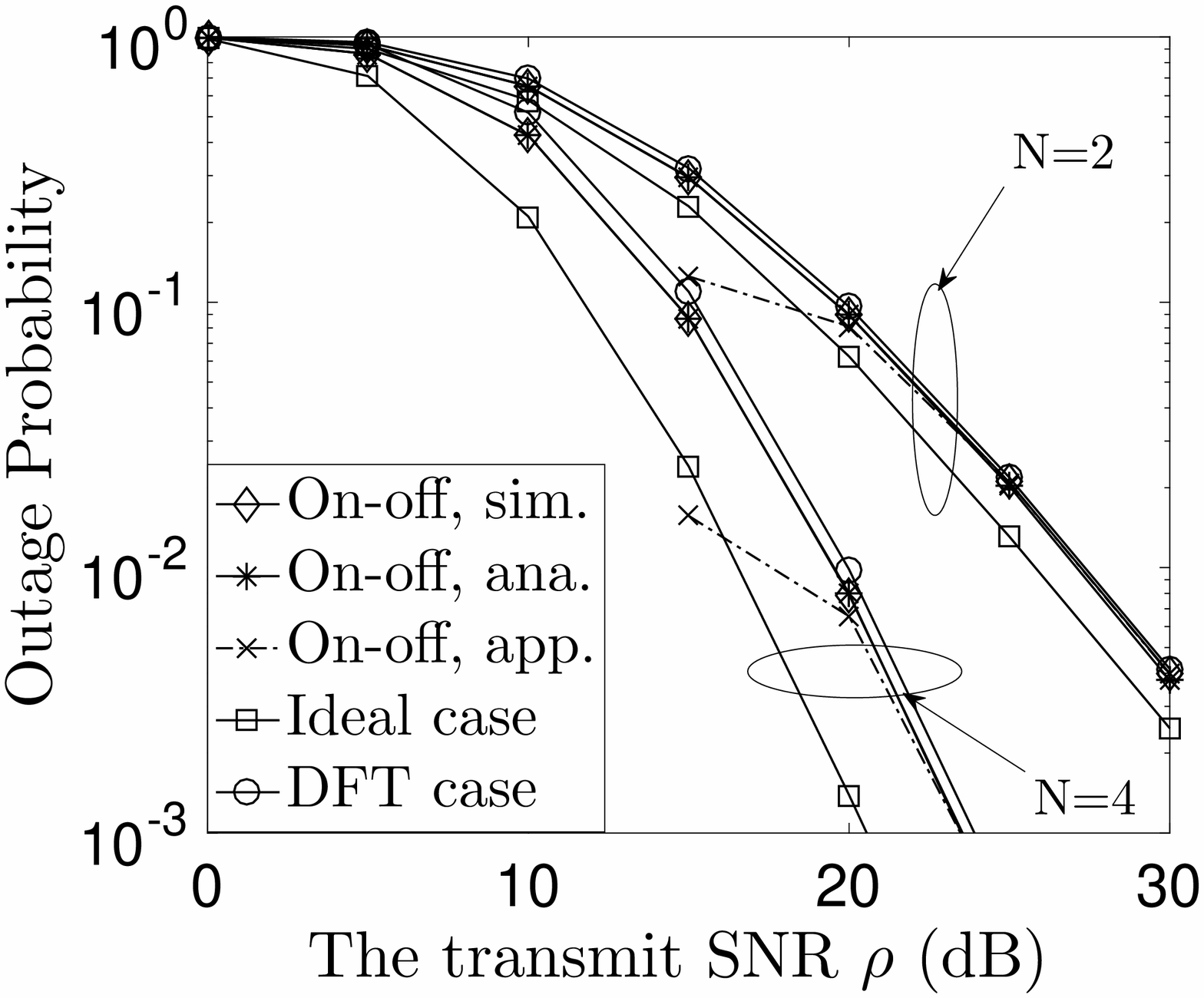}}
\subfigure[Impact of the choice of $Q$ on IRS-NOMA ($N=12$)]{\label{fig1b}\includegraphics[width=0.24\textwidth]{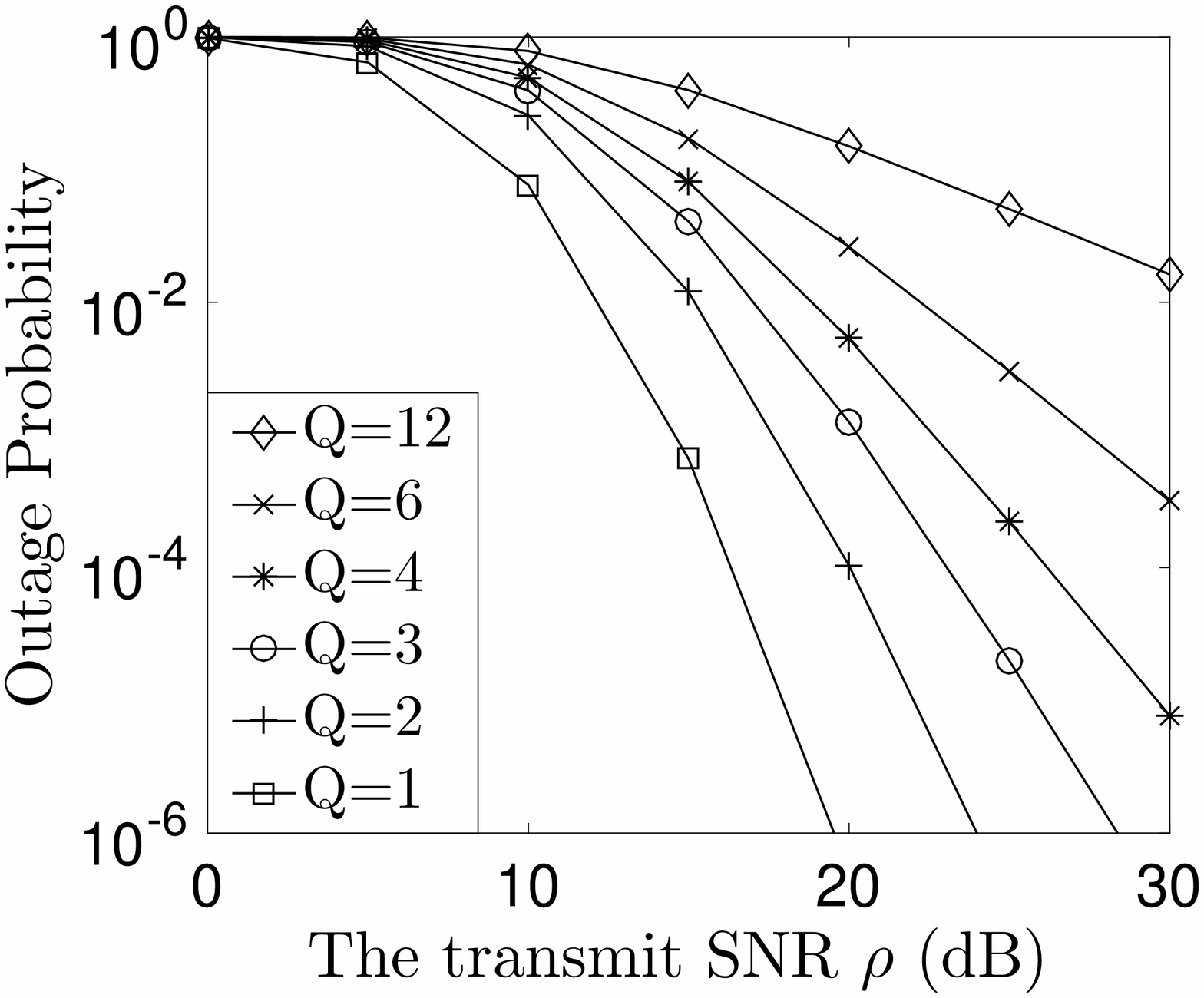}} \vspace{-0.5em}
\end{center}\vspace{-1em}
 \caption{Impact of IRS-NOMA on the downlink outage probability for the single-user case ($K=1$). $M=4$. $R_{k'}=2$ bits per channel use (BPCU). }\label{fig 1}\vspace{-1.5em}
\end{figure}

\vspace{-0.5em}
\section{Numerical Results}
In this section, computer simulation results are presented to demonstrate the performance of IRS-NOMA, where we use   $\alpha_1^2=\frac{4}{5}$ and $\alpha_2^2=\frac{1}{5}$. In Fig. \ref{fig 1}, the performance of IRS-NOMA is studied by focusing on the single-user case ($K=1$). Fig. \ref{fig1a} shows that  the slop of the outage probability curves for  the three schemes is the same, which indicates that they achieve the same diversity order. Among the three schemes, the one with ideal beamforming yields the best   performance, but  it might not be supported by  a practical antenna array. Among the two practical IRS-NOMA schemes, the one with on-off control yields better performance.  Fig. \ref{fig1a} also confirms the accuracy of the developed analytical results shown in Lemma \ref{lemma1}. Remark 1 indicates that increasing $Q$ decreases the achieved diversity gain, which is confirmed by Fig. \ref{fig1b}.

Fig. \ref{fig 2} shows the performance of the IRS-NOMA schemes when  there are multiple users ($K\geq 2$). Fig. \ref{fig2a} shows that for IRS-NOMA with ideal beamforming, the outage probability can be   reduced to zero by increasing the transmission power. However, there are   error floors for the two practical IRS-NOMA schemes. The reason for these error floors is due to the fact that the use of finite-resolution beamforming cannot  eliminate inter-pair interference completely. Fig. \ref{fig2a} also shows that the on-off scheme outperforms the DFT one, which is consistent to Fig. \ref{fig1a}.  The accuracy of the analytical results shown in Lemma \ref{lemma2} is confirmed by Fig. \ref{fig2b}.  In addition, Fig. \ref{fig2b} also shows  that increasing $N$ can effectively reduce the outage probability, as indicated in \eqref{20} in Lemma \ref{lemma2}. Fig. \ref{fig2c} confirms the optimality of the choice of $Q=1$ in the multi-user scenario, which is   consistent to Fig. \ref{fig1b}. 

\begin{figure}[t] \vspace{-2em}
\begin{center}\subfigure[Comparison between different IRS-NOMA ($N=4$ and $Q=1$)]{\label{fig2a}\includegraphics[width=0.24\textwidth]{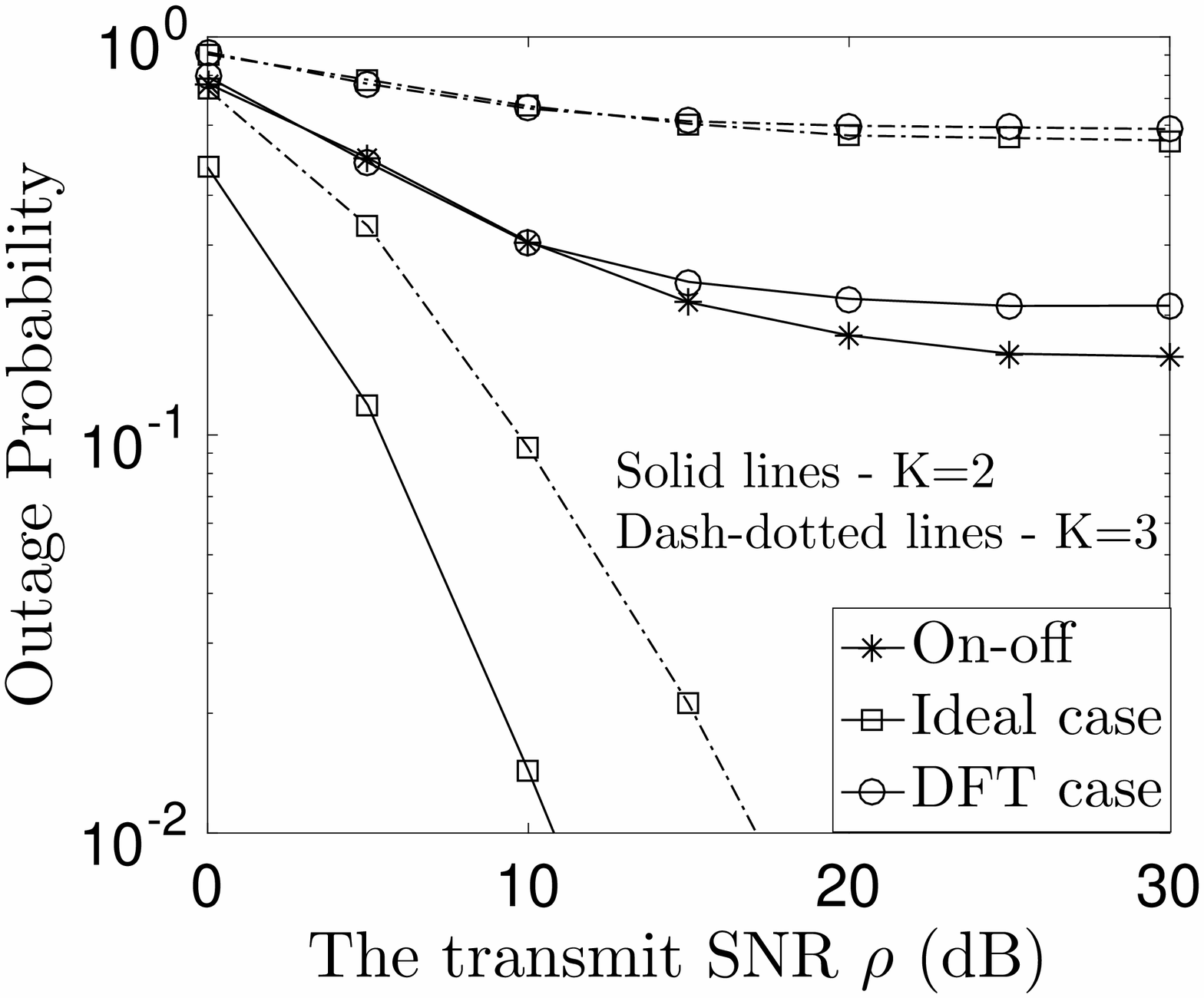}}
\subfigure[Effects of $N$ ($Q=1$) ]{\label{fig2b}\includegraphics[width=0.24\textwidth]{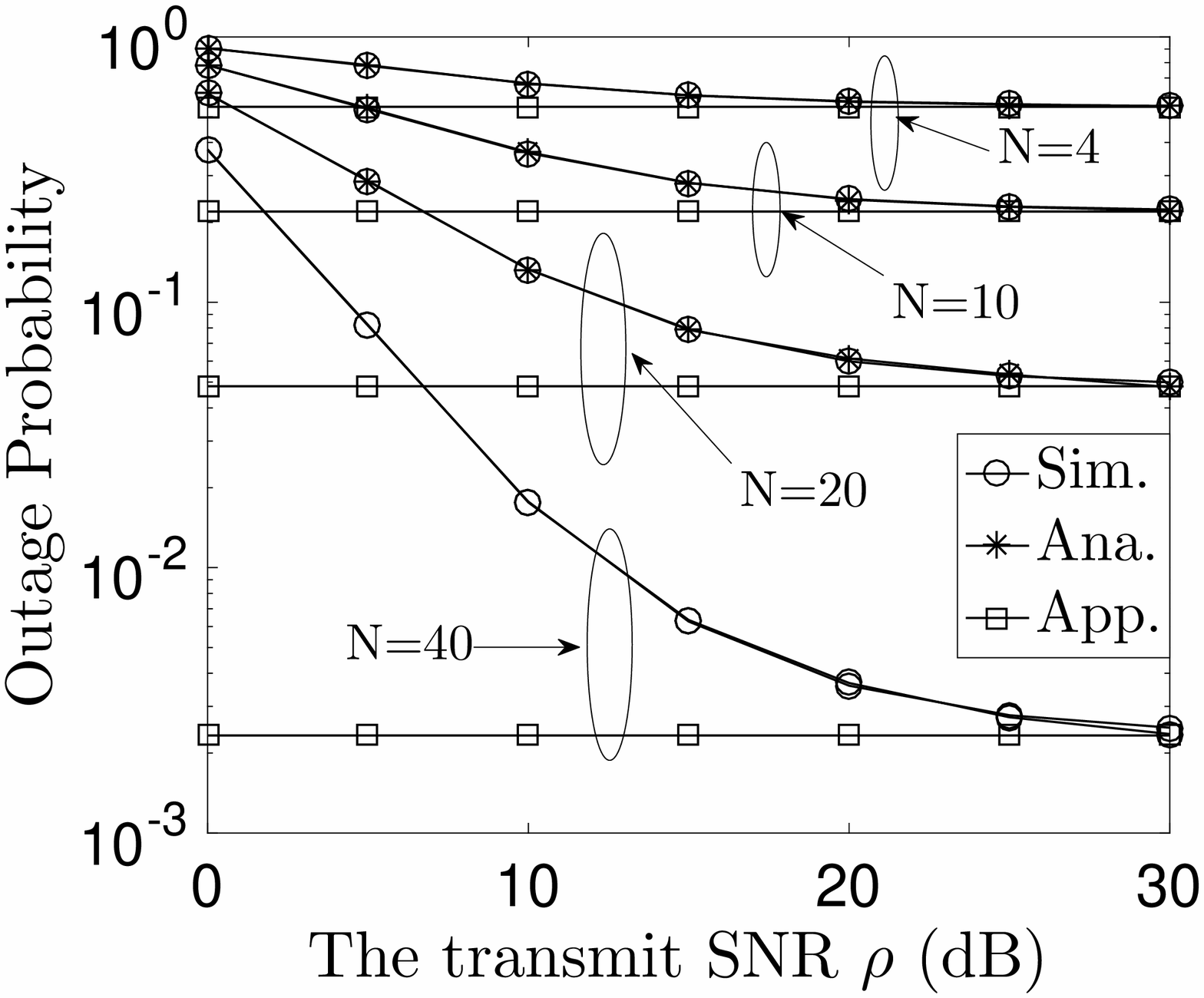}} \subfigure[Effects of $Q$ ($N=20$)]{\label{fig2c}\includegraphics[width=0.24\textwidth]{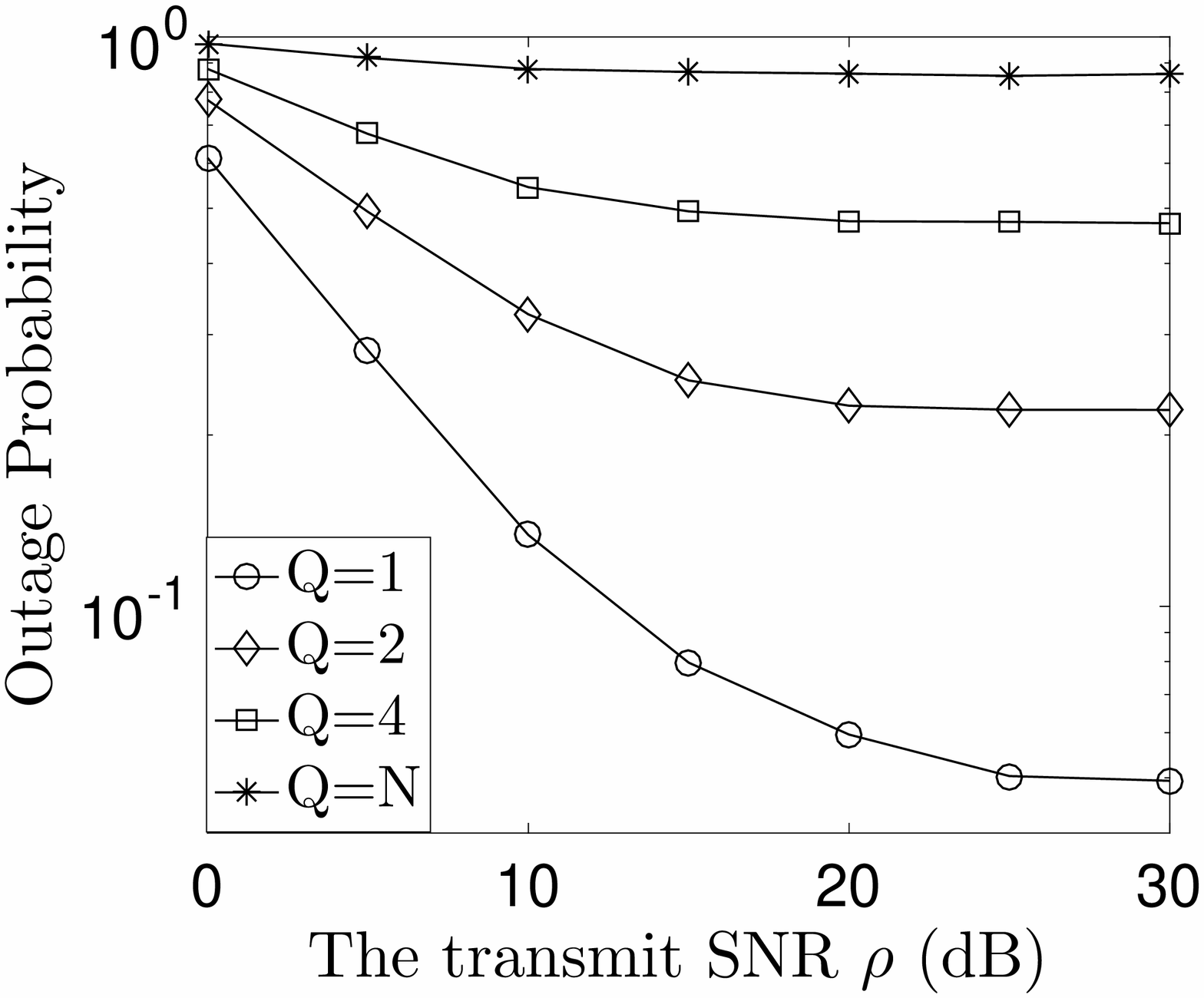}}\vspace{-0.5em}
\end{center}\vspace{-1em}
 \caption{Impact of IRS-NOMA on the downlink outage probability for the multi-user case ($K\geq 2$).  $M=4$. $R_{k'}=1$ BPCU.  }\label{fig 2}\vspace{-1.5em}
\end{figure}
\vspace{-0.5em}
\section{Conclusions}
  In this letter,   IRS-NOMA transmission has been proposed  to  ensure that more users can be served on each     orthogonal spatial direction, compared to SDMA. In addition, the impact of hardware impairments on the design of IRS-NOMA has been investigated and the performance of practical IRS-NOMA transmission has also been characterized. 

\vspace{-1em}
   \bibliographystyle{IEEEtran}
\bibliography{IEEEfull,trasfer}
   \end{document}